\documentclass[5p,authoryear,times]{elsarticle}

\usepackage{IEEEtrantools}
\usepackage[cmex10]{amsmath}
\usepackage{amsthm}
\usepackage{amssymb}
\usepackage{subfig}
\usepackage{amsfonts}
\usepackage{epstopdf}
\usepackage{cases}
\usepackage[vlined,ruled,commentsnumbered]{algorithm2e}

\newtheorem{thm}{Theorem}

\newtheorem{corollary}[thm]{Corollary}

\newtheorem{rmk}{Remark}

\begin{document}

\title{Enhanced Predictive Ratio Control of Interacting Systems}
\author[rvt]{Minh Hoang-Tuan Nguyen}
\ead{tuanminh@nus.edu.sg}

\author[rvt]{Kok Kiong Tan}\corref{cor1}
\ead{kktan@nus.edu.sg}

\author[rvt]{Sunan Huang}
\ead{elehsn@gmail.com}

\cortext[cor1]{Corresponding author}
\address[rvt]{National University of Singapore, Department of Electrical and Computer Engineering,\\3 Engineering Drive 3, Singapore 117576}
%
%

\begin{abstract}
Ratio control for two interacting processes is proposed with a PID feedforward design based on model
predictive control (MPC) scheme. At each sampling instant, the MPC control action minimizes a state-dependent performance index associated with a PID-type state vector, thus yielding a PID-type control structure. Compared to the standard MPC formulations with separated single-variable control, such a control action allows one to take into account the non-uniformity of the two process outputs. After reformulating the MPC control law as a PID control law, we provide conditions for prediction horizon and weighting matrices so that the closed-loop control is asymptotically stable, and show the effectiveness of the approach with simulation and experiment results.
\end{abstract}

\maketitle

\section{Introduction}\label{c6s1}
Ratio control has become a demanding task in industrial processes involving combustion systems or blending operations. Ratio control methods are used to maintain the flow rate of one stream in the process at a specified proportion relative to that of another (the wild flow). Besides the traditional series and parallel control, an alternative architecture, called Blend station \citep{Hag01Blend}, was proposed as auto-tuning and later improved in \cite{Vis05Design} for the choice of setpoint weighting. While ratio control of decoupled processes is well established, the problems become significantly complex for interacting processes. In this context, model predictive controllers (MPCs) have been recently applied to deal with ratio control, such as engine air–fuel and fuel–gas ratio control \citep{Gio06Hybrid,Mus08Adaptive,Suz09Individual}.

Among the various classes of MPCs, Generalized predictive control (GPC) is a potential method which overcome many pitfalls of other schemes when dealing with open loop unstable, nonminimum phase, or delayed systems \citep{Cla89Properties,Normey-Rico07Control}. Moreover, GPC can be used with multivariable systems by an model-augmented modification, even when constraints are considered. These advantages have been reviewed in \cite{Lee00Convergence,Bem02explicit}. Despite its efficiency, the computing burden discourages the widespread use of GPC compared to PID regulators in process industry. Compared with a true GPC method, PID control uses present and past data but not future information; moreover, its coefficients are limited to lower order polynomials than those of GPC law. To address GPC computational issues, several PID tuning procedures incorporating GPC were proposed so that they could achieve model-based control performance with a simpler structure. The idea of matching the GPC and PID control law structure was presented in \cite{Camacho03Model,Nes10Generalized,Sat10Design}. These papers showed that, by using a first/second-order system model, it is possible to simplify the GPC law as PID control law. A PID predictive controller was proposed in \cite{Mor03} where the author, rather than looking for the match of GPC and PID laws, considered a number of parallel PID controllers corresponding to the prediction horizon of GPC. In another context, the work in \cite{Tan00Development} developed a GPC-based PID controller by bringing PID error state into GPC performance index.

To bring these predictive PID design closer to the original ratio control problem, a previous work from \cite{Tan09Predictive} achieved composition control by changing setpoint when the output ratio is out of a predetermined threshold, without considering time delays. However, this setpoint variation method modifies control input through feedforward term outside MPC, so it easily upsets the input constraint. In addition, when the dead-time factor is included, especially different dead-times for individual processes, the information of future output ratio is demanded and the solution becomes more complicated. Thus the question is how to deal with a normal delayed process, as in \cite{Hag01Blend,Vis05Design}.

In this chapter, a PID feed-forward design based on predictive control concept is presented. It can be used for ratio control of two-input two-output (TITO) with inconsistent input delays. The solution for the delay case is solved by using equivalent control in MPC formula. Moreover, it incorporates ratio control into the performance index of GPC, so that no setpoint variation is required. The control law is still obtained as a feed-forward PID structure, with time-varying gains during the initial time-delay period and with constant gains thereafter. Proportion control is also taken care by a structural tuning. As a consequence, a feasible approach for proportion control is delivered.

The chapter is presented as follows. First, the state-space approach for TITO systems with dead-time is presented so that it includes the PID state vector (Section \ref{c6s2}). Second, the GPC control law is formulated in the given context, which allows us to recast GPC into feed-forward PID structure and criteria for choosing weighting matrices for the derived method are given so that the closed-loop system is asymptotically stable (Sections \ref{c6s3}). In Section \ref{c6s4}, the enhancement for ratio control through modification of the performance index is presented. Section \ref{c6s5} delivers simulation studies for the wafer thermal uniformity control example. Finally, experiment results are shown in Section \ref{c6s6} and the main principles of this chapter is concluded in Section \ref{c6s7}.

\subsection*{Notation}

For the examined system, $h$ denotes the input delay. The subscript $i\ (i=1,2)$ is to address the two channels of TITO systems. Besides, $r$ is the output setpoint, while $\tilde{r}$ is auxiliary reference and $\tilde{R}$ is the future auxiliary reference across prediction horizon. We also denote the system state as $X$ and PID state as $\tilde{X}$ in which $\theta$ is the integral term over output error $e$. Open-loop and closed-loop gains are indicated by $F$ and $\bar{F}$. The notation $Q>0\ (Q\geq0)$ denotes positive (semi) definiteness.

\section{State-space Representation of TITO System}\label{c6s2}

Consider the problem of regulating a process modeled by the typical FOPDT transfer functions:
\begin{IEEEeqnarray}{rCl}
y_1(z)&=&\frac{b_{11}z^{-h_1}}{z+a_{11}}u_1(z)+ \frac{b_{12}z^{-h_2}}{z+a_{12}}u_2(z)\nonumber\\
y_2(z)&=&\frac{b_{21}z^{-h_1}}{z+a_{21}}u_1(z)+ \frac{b_{22}z^{-h_2}}{z+a_{22}}u_2(z)
\label{c6eq1}
\end{IEEEeqnarray}
where $h_1\leq h_2\ (h_1,h_2\in \mathbb{R}^+)$ are input delays of the system. The output ratio between $y_1$ and $y_2$ is to be maintained at the desired value of $\alpha= \frac{r_2}{r_1}$ ($r_1,r_2$ are the output setpoints).

In order to deal with inconsistent input delays, we define the equivalent control as
\begin{equation}
U(k-h)=\begin{bmatrix}u_1(k-h_1)& u_2(k-h_2)\end{bmatrix}^T,
\label{c6eq2}
\end{equation}
used as a convenient notation for the derivation of MPC control law in Section \ref{c6s3}.

Rearrange \eqref{c6eq1} intro the difference equation and define special state definition $X(k)$ for TITO system (refer to \cite{Tan09Predictive} for details). By describing the PID state vector as $\tilde{X}_k = \begin{bmatrix}e_1(k)& e_1(k-1)& \theta_1(k)& e_2(k)& e_2(k-1)& \theta_2(k)\end{bmatrix}^T$, we have a complete state space equation
\begin{equation}
X(k+1)=FX(k)+GU(k-h)+E\tilde{r}(k),
\label{c6eq3}
\end{equation}
with 
\begin{equation}
X(k)=M \tilde{X}(k)+NU(k-1-h).
\label{c6eq4}
\end{equation}
These system matrices $F,G,E,M,N$ are given as
\begin{IEEEeqnarray}{rCl}
F&=&\begin{bmatrix}-a_{11}+a_{12}& 1& 0& 0& 0& 0\\
									-a_{11}a_{12}& 0& 0& 0& 0& 0\\
									1& 0& 1& 0& 0& 0\\
									0& 0& 0& -a_{21}+a_{22}& 1& 0\\
									0& 0& 0& -a_{21}a_{22}& 0& 0\\
									0& 0& 0& 1& 0& 1\end{bmatrix},\\
G&=&\begin{bmatrix}-b_{11}& -b_{12}\\
									-b_{11}a_{12}& -b_{12}a_{11}\\
									0& 0\\
									-b_{21}& -b_{22}\\
									-b_{21}a_{22}& -b_{22}a_{21}\\
									0& 0\end{bmatrix},\nonumber\ 
E=\begin{bmatrix}1& 0\\0& 0\\0& 0\\0& 1\\0& 0\\ 0& 0\end{bmatrix},\\
F&=&\begin{bmatrix}1& 0& 0& 0& 0& 0\\
									0& -a_{11}a_{12}& 0& 0& 0& 0\\
									0& 0& 1& 0& 0& 0\\
									0& 0& 0& 1& 0& 0\\
									0& 0& 0& 0& -a_{21}a_{22}& 0\\
									0& 0& 0& 0& 0& 1\end{bmatrix},\ 
N=\begin{bmatrix}0& 0\\
									-b_{11}a_{12}& -b_{12}a_{11}\\
									0& 0\\
									0& 0\\
									-b_{21}a_{22}& -b_{22}a_{21}\\
									0& 0\end{bmatrix}.\nonumber
\end{IEEEeqnarray}

\section{Predictive PID controller}\label{c6s3}
\subsection{GPC Control Law}
The system model is written as
\begin{equation}
X(k+1)=FX(k)+GU(k-h)+E \tilde{r}(k)
\label{c6eq5}
\end{equation}
where $X\in \mathbb{R}^n,\ U\in \mathbb{R}^m\ (n=6,m=2)$. With this model, the following problem is posed: given the current state $X(k)$, find the equivalent $N$-step control sequence $\bar{U}=\{U(k-h),U(k-h+1),...,U(k-h+N-1)\}$ that minimizes the performance index:
\begin{equation}
J=\sum_{j=k}^{k+N-1}[X(j+1)^TQ_jX(j+1)+U(j-h+1)^TR_jU(j-h+1)].
\label{c6eq6}
\end{equation}
In \eqref{c6eq6}, $N$ is the prediction horizon; $Q_j\geq 0,\ R_j>0$ are the state and control weighting matrices.

Now define stacked vectors $\bar{X}=\begin{bmatrix}X(k+1)& ...& X(k+N)\end{bmatrix}^T$, $\tilde{R}(k)=[\tilde{r}(k)\,\ldots\,\tilde{r}(k+N-1)]^T$. Then \eqref{c6eq5} can be written as
\begin{equation}
\bar{X}= HFX(k)+P\bar{U}+\bar{E}\tilde{R}(k),
\label{c6eq7}
\end{equation}
where
\begin{IEEEeqnarray}{rCl}
H&=&\begin{bmatrix}I\\ F\\ ...\\ F^{l-1}\end{bmatrix},
P = \begin{bmatrix}G& 0& ...& 0\\
									FG& G& ...& 0\\
									...& ...& ...& ...\\
									F^{l-1}G& F^{l-2}G& ...& G\end{bmatrix},\nonumber\\
\bar{E} &=& \begin{bmatrix}E& 0& ...& 0\\
									FE& E& ...& 0\\
									...& ...& ...& ...\\
									F^{l-1}E& F^{l-2}E& ...& E\end{bmatrix}.\nonumber
\end{IEEEeqnarray}

By doing so, the performance index \eqref{c6eq6} can be expressed as
\begin{equation}
J = \bar{X}^TQ\bar{X} + \bar{U}^TR\bar{U}.
\label{c6eq8}
\end{equation}
The corresponding optimal control law is determined by taking the gradient $\partial J/ \partial \bar{U}$ to be zero, so that
\begin{equation}
\bar{U}=-(P^TQP+R)^{-1}P^TQ(HFX(k)+\bar{E}\tilde{R}(k)).
\label{c6eq9}
\end{equation}
Apply the receding horizon control concept, the first-step input is
\begin{IEEEeqnarray}{rCl}
U(k-h)&=&-D(P^TQP+R)^{-1}P^TQ(HFX(k)+\bar{E}\tilde{R}(k))\nonumber\\
&=& K_{GPC}X(k)+K_{ref}\tilde{R}(k),
\label{c6eq10}
\end{IEEEeqnarray}
where $D=\begin{bmatrix}1& 0& ...& 0\end{bmatrix}$, $K_{GPC}=-D(P^TQP+R)^{-1}P^TQHF=\begin{bmatrix}K_{1GPC}& K_{2GPC}\end{bmatrix}^T$ and $K_{ref}=D(P^TQP+R)^{-1}P^TQ\bar{E}=[K_{1ref}\ K_{2ref}]^T$. The second term in \eqref{c6eq10} can be considered as a feed-forward part of the controller design, assuming that the future setpoint sequence is known. It follows from the equivalent control definition in \eqref{c6eq2} that
\begin{IEEEeqnarray}{rCl}
u_1(k)=K_{1GPC}X(k+h_1)+K_{1ref}\tilde{R}(k+h_1)\nonumber\\
u_2(k)=K_{2GPC}X(k+h_2)+K_{2ref}\tilde{R}(k+h_2).
\label{c6eq11}
\end{IEEEeqnarray}

\subsection{Future State Prediction}
From \eqref{c6eq11}, it can be seen that in order to minimize $J$, the control at the current instant depends on the fixed gains $K_{GPC}$ and a future state at time $k+h_1$ and $k+h_2$. 

\subsubsection{For $k>h_2$}

Let $\bar{F}=F+GK_{GPC}$. In order to predict the future states, a closed-loop equation is formed by combining \eqref{c6eq3} and \eqref{c6eq10}:
\begin{equation}
X(k+1)=\bar{F}X(k)+GK_{ref}\tilde{R}(k)+E\tilde{r}(k).
\label{c6eq12}
\end{equation}
From the one-step prediction above, the future states $X(k+h_1),\ X(k+h_2)$ are determined iteratively by 
\begin{IEEEeqnarray}{rCl}
X(k+h_1)=&&\bar{F}^{h_1}X(k)+\bar{F}^{h_1-1}[GK_{ref}\tilde{R}(k)+E\tilde{r}(k)]+ ...\nonumber\\ &+&[GK_{ref}\tilde{R}(k+h_1-1)+E \tilde{r}(k+h_1-1)],\label{c6eq13}\\
X(k+h_2)=&&\bar{F}^{h_2}X(k)+\bar{F}^{h_2-1}[GK_{ref}\tilde{R}(k)+E\tilde{r}(k)]+ ...\nonumber\\ &+&[GK_{ref}\tilde{R}(k+h_2-1)+E \tilde{r}(k+h_2-1)].
\label{c6eq14}
\end{IEEEeqnarray}
As seen from \eqref{c6eq13}, \eqref{c6eq14}, the coefficient of $X(k)$ in these formula is independent of time $k$ for $k>h_2$. In the next case, we will see that the state prediction during time-delay period has the $k$-dependent gains.

\subsubsection{For $1\leq k \leq h_2$}

Denote $\bar{F^1}=F+G\begin{bmatrix}K_{1GPC}& 0\end{bmatrix}^T$, $K_{ref}^1=\begin{bmatrix}K_{1ref}& 0\end{bmatrix}^T$, with the superscript $(.)^1$ indicating the region $min{h_1,h_2}<k<max{h_1,h_2}$. Depending on the existence of the optimal input in \eqref{c6eq10}, the system in \eqref{c6eq5} can become
\begin{numcases}{X(l+1)=}
FX(l) & if $l\leq h_1$ \nonumber \\
\bar{F^1}X(l)+K_{ref}^1\tilde{R}(k) + E\tilde{r}(k) & if $h_1\leq l\leq h_2$\IEEEeqnarraynumspace\\
\bar{F}X(l)+GK_{ref}\tilde{R}(k) + E\tilde{r}(k) & if $l\geq h_2$\nonumber.
\label{c6eq15}
\end{numcases}
Now $l$ can be substituted by $k+h_1$ or $k+h_2$ to get the future states.

\subsection{Predictive PID Control Law}\label{sec623}
Substituting the predicted states obtained in \eqref{c6eq13}, \eqref{c6eq14} into the control law \eqref{c6eq11}
\begin{IEEEeqnarray}{rCl}
u_1(k) &=& K_{1GPC}\bar{F}_1X(k)+S_1(k)\nonumber\\
u_2(k) &=& K_{2GPC}\bar{F}_2X(k)+S_2(k),
\label{c6eq16}
\end{IEEEeqnarray}
where $\bar{F}_1,\bar{F}_2$ are the coefficients associated with $X(k)$ and $S_1(k)$, $S_2(k)$ are the terms that involve future reference. $S_1(k),S_2(k)$ can be updated at every step, as in the Algorithm 1 below.

The control law in \eqref{c6eq16} can be incorporated within the PID structure by using \eqref{c6eq4}:
\begin{IEEEeqnarray}{rCl}
u_1(k)&=&K_{1PID}\tilde{X}(k)+K_{1u}U(k-1-h)+S_1(k)\nonumber\\
u_2(k)&=&K_{2PID}\tilde{X}(k)+K_{2u}U(k-1-h)+S_2(k),
\label{c6eq17}
\end{IEEEeqnarray}
where $K_{1PID}=K_{1GPC}\bar{F}_1M$, $K_{2PID}=K_{2GPC}\bar{F}_2M$ and $K_{1u}=K_{1GPC}\bar{F}_1 N$, $K_{2u}=K_{2GPC}\bar{F}_2N$.

\begin{rmk}
In Eq. \eqref{c6eq17} each of the control inputs is navigated by the outputs of two PIDs (as $\tilde{X} \in \mathbb{R}^6$) and a feed-forward term that consists of the rest of the formula.
\end{rmk}

As one observes, the MPC law based on future output prediction in \eqref{c6eq11}, which is open-loop in nature, has been reformed to a closed-loop control law as in \eqref{c6eq17}. The closed-loop stability would be guaranteed later on Section 2.3. It is also worth mentioning that because of the future state prediction during time-delay period $max{h_1,h_2}=h_2$, this PID formulation has time-varying gains during initial stage . Beyond this period, the PID controller resumes constant gains. In general, the state feedback control law \eqref{c6eq17} refers to the optimal lookup table for the PID gains, and a closed-form solution is created.

The predictive PID algorithm can be summarized in the following:

\begin{algorithm}[ht]
\SetAlgoLined
\KwData{$k$, $\tilde{r}$, $X$}
\KwResult{$K_{1PID},\ K_{2PID},\ K_{1u},\ K_{2u}$}
initialize $\tilde{R}(k),\ \tilde{R}(k+h_1),\ \tilde{R}(k+h_2)$ by definition in \eqref{c6eq7}. Determine $K_{GPC}$ and $K_{ref}$ offline from \eqref{c6eq10}.
\eIf{$k\leq h_2$}{
   $S\leftarrow 0$, $F_b\leftarrow I$, $\tilde{R}\leftarrow \tilde{R}(k)$\;
   \For{$i\leftarrow k+h_2-1$}{
   		update $\tilde{R}$ by removing $\tilde{r}(i)$ and adding $\tilde{r}(i+N)$ to the queue\;
   		$\tilde{r}\leftarrow \tilde{r}(i+1)$\;
   		Assign\\
				\Indp
				$S\leftarrow FS+E\tilde{r},\ F_b\leftarrow FF_b$ if $i\leq h_1$\;
				$S\leftarrow \bar{F^1}S+GK^1_{ref}\tilde{R}+E\tilde{r},\ F_b\leftarrow \bar{F^1}F_b$ if $h_1\leq i\leq h_2$
				$S\leftarrow \bar{F}S+GK_{ref}\tilde{R}+E\tilde{r},\ F_b\leftarrow \bar{F}F_b$ if $i\geq h_2$\;
				\Indm
   		\If{$i==k+h_1-1$}{
   			$S_1\leftarrow K_{1GPC}S+K_{1ref}\tilde{R}(k+h_1)$\;
   			$\bar{F}_1\leftarrow F_b$\;
   		}
   }
   $S_2\leftarrow K_{2GPC}S+K_{2ref}\tilde{R}(k+h_2)$\;
   $\bar{F}_2\leftarrow F_b$\;
   evaluate the gains $K_{1PID},\ K_{2PID},\ K_{1u},\ K_{2u}$ from \eqref{c6eq17}\;
}{
   Fix $K_{1PID},\ K_{2PID},\ K_{1u},\ K_{2u}$ from here on\;
}
\caption{Computation of predictive PID gains.}
\label{c6alg1}
\end{algorithm}

\subsection{Stability}
As the system has time delays incorporated in its transfer functions, the stability criterion becomes more complex than the one suggested in the work of [15]. The closed-loop stability
created by the proposed feedback is analyzed in long-term situation where the PID controllers have already passed the initial stage of delay and converged to the fixed gain region (k > h2). Without
loss of generality, all reference values are assumed to be zero, and the dead-time $h_2\geq h_1$. From \eqref{c6eq16},
\begin{IEEEeqnarray}{rCl}
U(k-h)&=&\begin{bmatrix}u_1(k-h_1)\\ u_2(k-h_2)\end{bmatrix}=\begin{bmatrix}K_{1GPC}\bar{F}^{h_1}X(k-h_1)\\ K_{2GPC}\bar{F}^{h_2}X(k-h_2)\end{bmatrix}\nonumber\\
&=&\begin{bmatrix}K_{1GPC}\\ 0\end{bmatrix}\bar{F}^{h_1}X(k-h_1)+\begin{bmatrix}0\\ K_{2GPC}\end{bmatrix}\bar{F}^{h_2}X(k-h_2)\nonumber\\
&=&K^1_{GPC}\bar{F}^{h_1}X(k-h_1)+K^2_{GPC}\bar{F}^{h_2}X(k-h_2),
\label{c6eq18}
\end{IEEEeqnarray}
where $K^1_{GPC}=\begin{bmatrix}K_{1GPC}& 0\end{bmatrix}^T$, $K^2_{GPC}=\begin{bmatrix}0& K_{2GPC}\end{bmatrix}^T$. Substituting \eqref{c6eq18} into \eqref{c6eq3}, we obtain
\begin{equation}
X(k+1)=FX(k)+GK^1_{GPC}\bar{F}^{h_1}X(k-h_1)+GK^2_{GPC}\bar{F}^{h_2}X(k-h_2).
\label{c6eq19}
\end{equation}

Now, the stability condition of the closed-loop system \eqref{c6eq19} is presented through Theorem 1.
\begin{thm}
The system \eqref{c6eq19} will be stable if and only if all the roots $\lambda$ of the following determinant equation
\begin{equation}
det [\lambda^{h_2+1}I-F\lambda^{h_2}-GK^1_{GPC}\bar{F}^{h_1}\lambda^{h2-h1}-GK^2_{GPC}\bar{F}^{h_2}]=0,
\label{c6eq20}
\end{equation}
satisfy $\left|\lambda\right|<1$, assuming that $h_2\geq h_1$.
\end{thm}
\begin{proof}
From \eqref{c6eq19}, a new state space equation is constructed as
\begin{IEEEeqnarray}{rCl}
\begin{bmatrix}X(k-h_2+1)\\ \vdots \\X(k)\\X(k+1)\end{bmatrix} &=& \begin{bmatrix}0& I & & \ldots & & 0\\
& & \ddots & & & \\
\vdots & \vdots & & I & & \vdots \\
& & & & \ddots & \\
0& 0& & \cdots & & I\\
GK^2_{GPC}\bar{F}^{h_2}& 0& \ldots & GK^1_{GPC}\bar{F}^{h_1} & \ldots & F
\end{bmatrix}\nonumber\\
&&.\begin{bmatrix}X(k-h_2)\\...\\X(k-1)\\X(k)\end{bmatrix}
\end{IEEEeqnarray}
This is the canonical controllable block form, in which the characteristic equation is obtained easily. The proof is directly followed by a block elimination which leads to lower triangular block form, as in the singular form. Thus the above system has eigenvalues which are obtained by solving the equation
\begin{equation}
det [\lambda^{h_2+1}I-F\lambda^{h_2}-GK^1_{GPC}\bar{F}^{h_1}\lambda^{h2-h1}-GK^2_{GPC}\bar{F}^{h_2}]=0.\nonumber
\end{equation}
Therefore, this system will be asymptotically stable if all the eigenvalues are within the unit circle, or the condition of \eqref{c6eq20} to be satisfied. Note that the size of the matrix $[\lambda^{h_2+1}I-F\lambda^{h_2}-GK^1_{GPC}\bar{F}^{h_1}\lambda^{h2-h1}-GK^2_{GPC}\bar{F}^{h_2}]$ is equal to $6 \times 6$. Interested readers are referred to \cite{Sai66control} for further detail on determinant
equation which helps to reduce the size of the matrix when larger systems are concerned.
\end{proof}

\begin{figure}[ht]%
\centering
\includegraphics[width=\columnwidth]{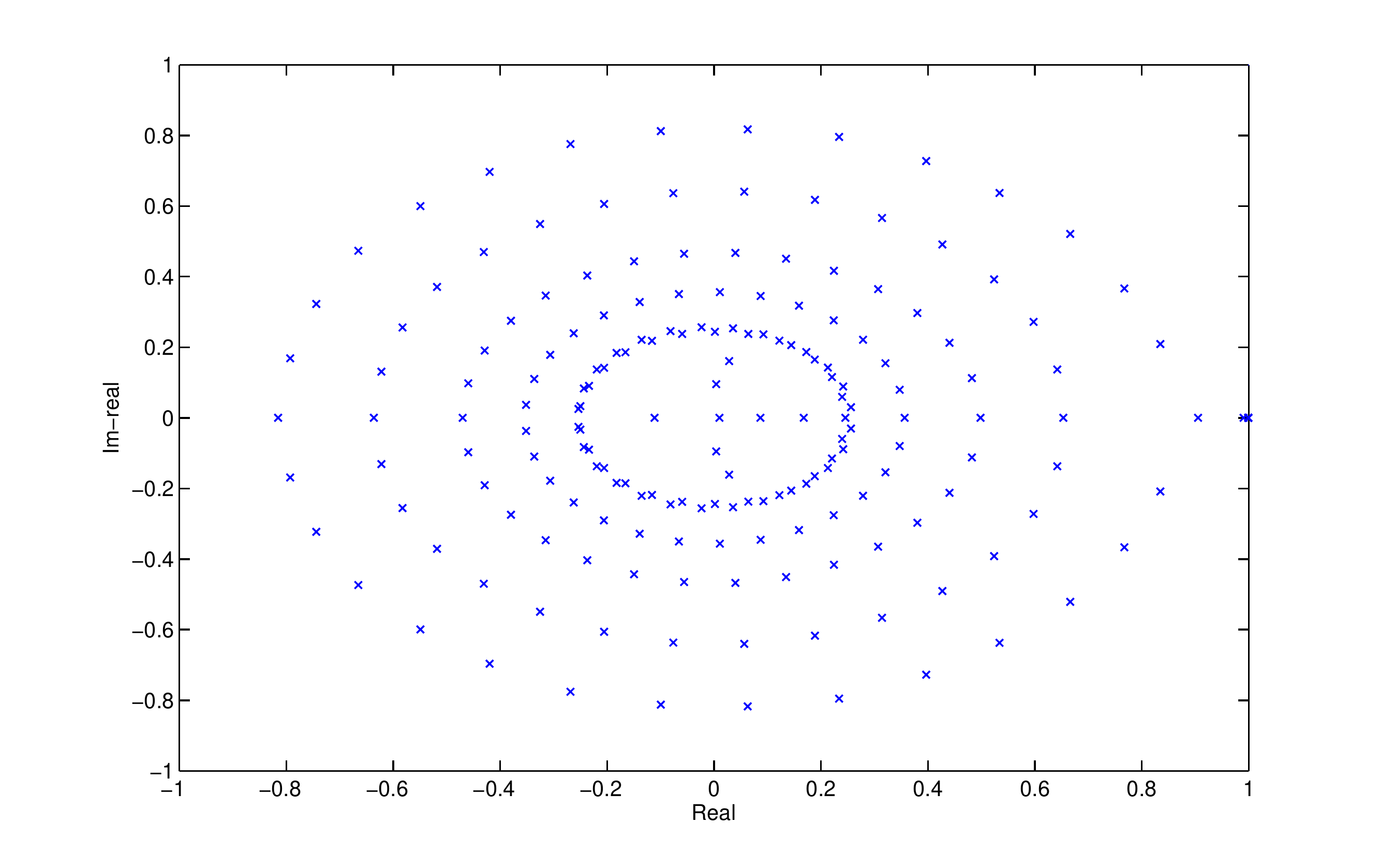}%
\caption{A typical eigenvalue map of the closed-loop system using the proposed method.}%
\label{c6fig2}%
\end{figure}

\begin{corollary}
The condition in \eqref{c6eq20} implies a necessary condition that all eigenvalues of the matrix $\bar{F} = F + GK_{GPC}$ is within the unit circle.
\end{corollary}
\begin{proof}
Indeed, note that $K^1_{GPC}+K^2_{GPC}=K_{GPC}$ and $\bar{F}=F+GK_{GPC}$, so
\begin{IEEEeqnarray}{rCl}
&&\lambda^{h_2+1} I-F\lambda^{h_2}-GK_{GPC}^1 \bar{F} \lambda^{h_1} \lambda^{h_2-h_1}-GK_{GPC}^2 \bar{F}^{h_2}\nonumber\\
&&=\lambda^{h_2+1} I-[\bar{F}-G(K_1+K_2 )]\lambda^{h_2}-GK_{GPC}^1 \bar{F}^{h_1} \lambda^{h_2-h_1} -GK_{GPC}^2 \bar{F}^{h_2}\nonumber\\
&&=(\lambda^{h_2+1} I-\bar{F}\lambda^{h_2})+GK_{GPC}^1 \lambda^{h_2-h_1} (\lambda^{h_1} I-\bar{F}^{h_1} )+GK_{GPC}^2 (\lambda^{h_2} I-\bar{F}^{h_2})\nonumber\\
&&=[\lambda^{h_2}+GK_GPC^1 (\lambda^{h_2-1} I+\bar{F}\lambda^{h_2-2}+\ldots+\bar{F}^{h_1-1} \lambda^{h_2-h_1} )\nonumber\\
&&\quad +GK_GPC^2 (\lambda^{h_2-1} I+\bar{F}\lambda^{h_2-2}+\ldots+\bar{F}^{h_2-1} )](\lambda I-\bar{F} ).\IEEEeqnarraynumspace
\label{c6eq21}
\end{IEEEeqnarray}
Thus, the eigenvalues of $\bar{F}$ must be within unit circle in order to satisfy \eqref{c6eq20}. This
is consistent with the result obtained in [15]. A typical eigenvalue map for the system \eqref{c6eq19} is presented in Fig. \ref{c6fig2}.
\end{proof}

\section{Tightening ratio control}\label{c6s4}
\subsection{Ratio control design}
Ratio control, traditionally, is implemented either via a series configuration with $r_2 =\alpha y_1$ or a parallel one with $r_2 =\alpha r_1$. The parallel configuration proves to be better than series configuration in removing or reducing lag phenomenon of slave variable. However, it incurs a different disadvantage, an open-loop design, in which a significant upset to the ratio of the variables can follow
when a large or fast load disturbance occurs, which cannot be tolerated in certain applications such as the wafer temperature uniformity control. Hence, the setpoint variation scheme was proposed in [15]. The dynamic information of ratio error was reflected in setpoint and it adjusts the optimal control law in \eqref{c6eq17} through feed-forward calculation. This can only be applied for systems with no delay, since threshold decision and ratio error in future time after the delay may be difficult to predict.

A new ratio control scheme is proposed, which can also improve the transient performance and disturbance rejection. The first advantage over setpoint variation is that the prediction of future ratio error is avoided. Moreover, this scheme is imposed directly into the performance index, thus achieving optimal control through PID gains instead of feedforward control. This is implemented by introducing the error ratio into the performance index $J$.

Let us fraction Q into $Q_1+\beta Q_2+\gamma Q_3\ (\beta,\gamma \in \mathbb{R}^+)$ where $\beta,\gamma$ are weighting factors. For simplicity, define $Q_1$ as an identity matrix; this matrix would be used as a normal gain for output tracking. Besides, define $Q_2=M_2^T M_2$ and $Q_3=M_3^TM_3$ such that
\begin{IEEEeqnarray}{rCl}
M_2&=&\begin{bmatrix}1& 0& 0 -\frac{1}{\alpha}& 0& 0\end{bmatrix}\nonumber\\
M_3&=&\begin{bmatrix}0& 0& 1& 0& 0& -\frac{1}{\alpha}\end{bmatrix}\nonumber
\end{IEEEeqnarray}

With the definition of the system state $X(k)$ in Section 2.1, it follows that
\begin{IEEEeqnarray}{rCl}
\left\|X(k)\right\|^2_{Q_2}&=&[M_2 X(k)]^T.[M_2 X(k)]=(e_1(k)-\frac{1}{\alpha}e_2(k))^2,\nonumber\\
\left\|X(k)\right\|^2_{Q_3}&=&[M_3X(k)]^T.[M_3X(k)]=(\theta_1(k)-\frac{1}{\alpha}\theta_2(k))^2,\nonumber
\end{IEEEeqnarray}
and these two terms could be used to optimize the output ratio effectively.

The role of $Q_2$ is to control the output errors $e_1=r_1-y_1$, $e_2=r_2-y_2$ to follow the desired output ratio $\alpha$. Normally, the term $Q_1$ commands the two processes outputs  $y_1 (k)$ and $y_2 (k)$ to the setpoints $r_1$, $r_2$ without taking care of the ratio $y_2/y_1$ during the transient stage. Since one knows the information $r_2/r_1=\alpha$, controlling the error ratio $e_2/e_1$ towards $\alpha$ can be an advantage in assuring the desired output ratio. The attractive point is that this feature still works when the initial output ratio is different from the desired output ratio, or, the ratio setpoint $\alpha$ is varying.

\begin{figure}%
\centering
\includegraphics[width=\columnwidth,height=2.5in]{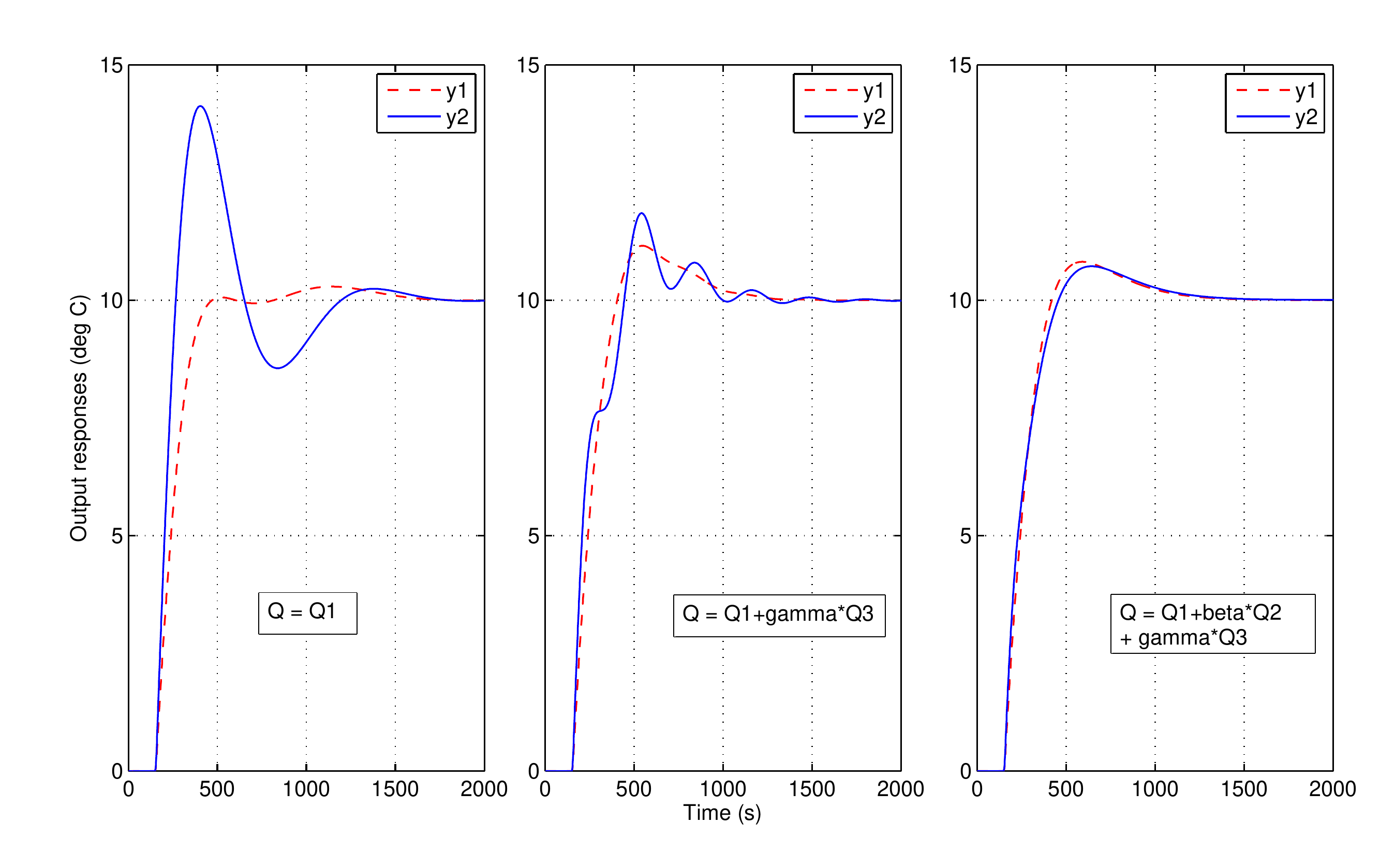}%
\caption{Tuning for weighting parameters $\beta$ and $\gamma$.}%
\label{c6fig3}%
\end{figure}

If one considers $Q_2$ as the proportional gain for ratio error, then $Q_3$ plays the role of integral gain. It helps to shape the response rates of the two processes to be closer to each other, instead of force the faster flow to following the slower one. In other words, the output ratio returns to the desired value faster and is prevented from possible offset. This can be illustrated in Fig. \ref{c6fig3}.

As a whole, the new performance index would be changed to
\begin{equation}
J=\sum_{k=t+i}^{t+N}(\left\|X(k)\right\|^2_{Q_1+\beta Q_2+\gamma Q_3}+\left\|U(k-h)\right\|^2_R),
\label{c6eq22}
\end{equation}
dependent on the balance of $Q_1$ (output error), $Q_2$ (ratio error) and $Q_3$ (ratio error integrator). A tuning method for $\beta,\ \gamma$ will be discussed more in the next section.

Again, since the ratio dynamic information is used as feedback within the performance index, disadvantages such as lag phenomenon and open-loop problem, caused by the traditional designs, could be reduced for the most part.

\begin{rmk}
This systematic tuning for $Q$ in \eqref{c6eq18} is more adequate than the arbitrary tuning in \eqref{c6eq6}. As this algorithm focuses on reduces the ratio error while driving outputs to the setpoints, weighting factors are put among $Q_1$ (output error), $Q_2$ (ratio error) and $Q_3$ (ratio error integrator) to balance the priority of these goals. It is also easier for practical users to decide the positive real values of $\beta$ and $\gamma$ rather than the original matrix $Q$, which is usually chosen in diagonal form.
\end{rmk}

\subsection{Tuning weighting matrices}
A formal tuning procedure for the new ratio controller proposed in Section 3.1 must satisfy the stability condition in Section 2.3. In this part, an simple, practical tuning method is presented.

\begin{figure}[!ht]
\centering
\subfloat{\includegraphics[width=\columnwidth, trim=0cm 1cm 0cm 1cm, height = 1.8in]{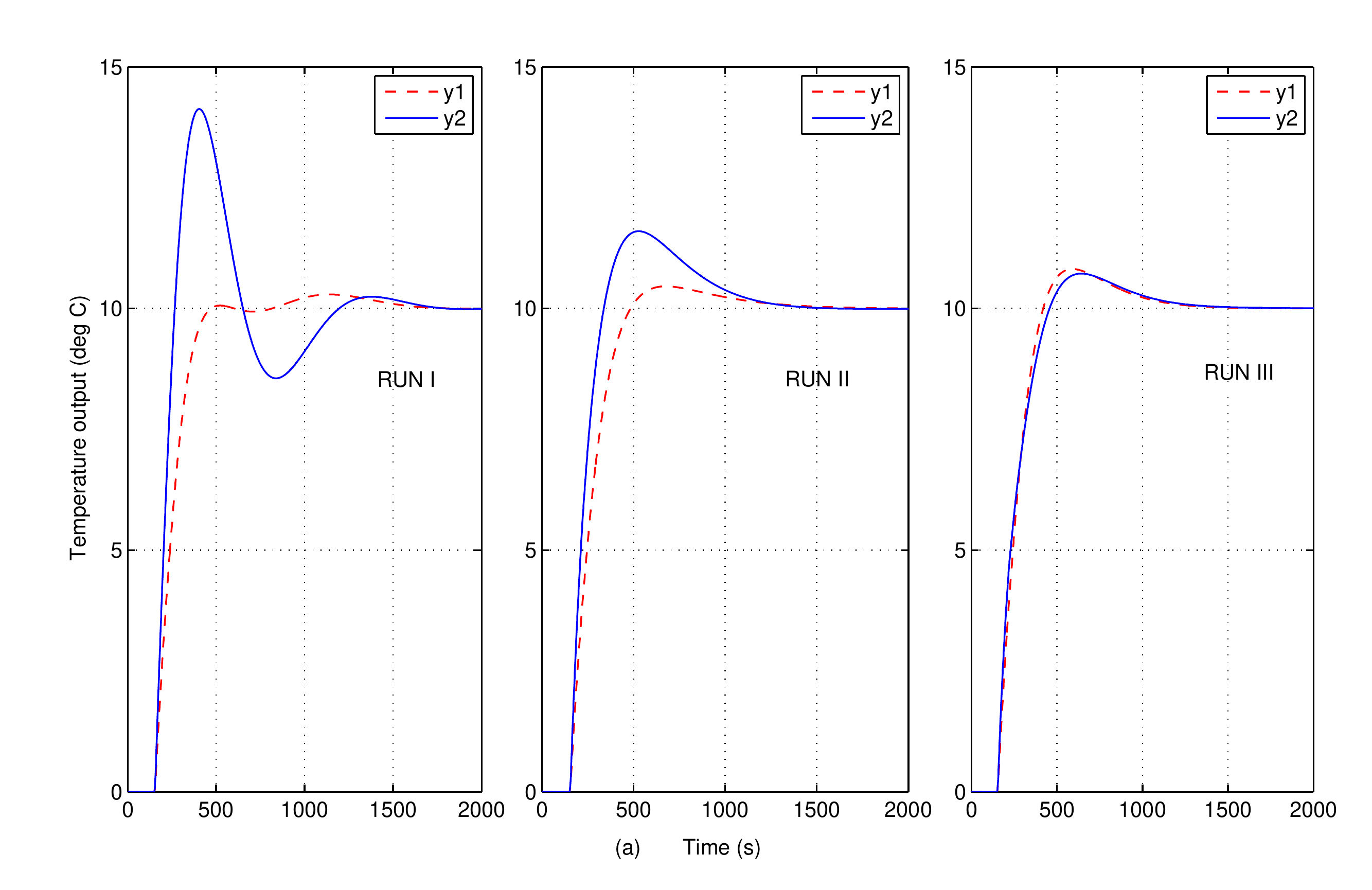}}\\
\subfloat{\includegraphics[width=\columnwidth, trim=0cm 1cm 0cm 0cm, height = 1.8in]{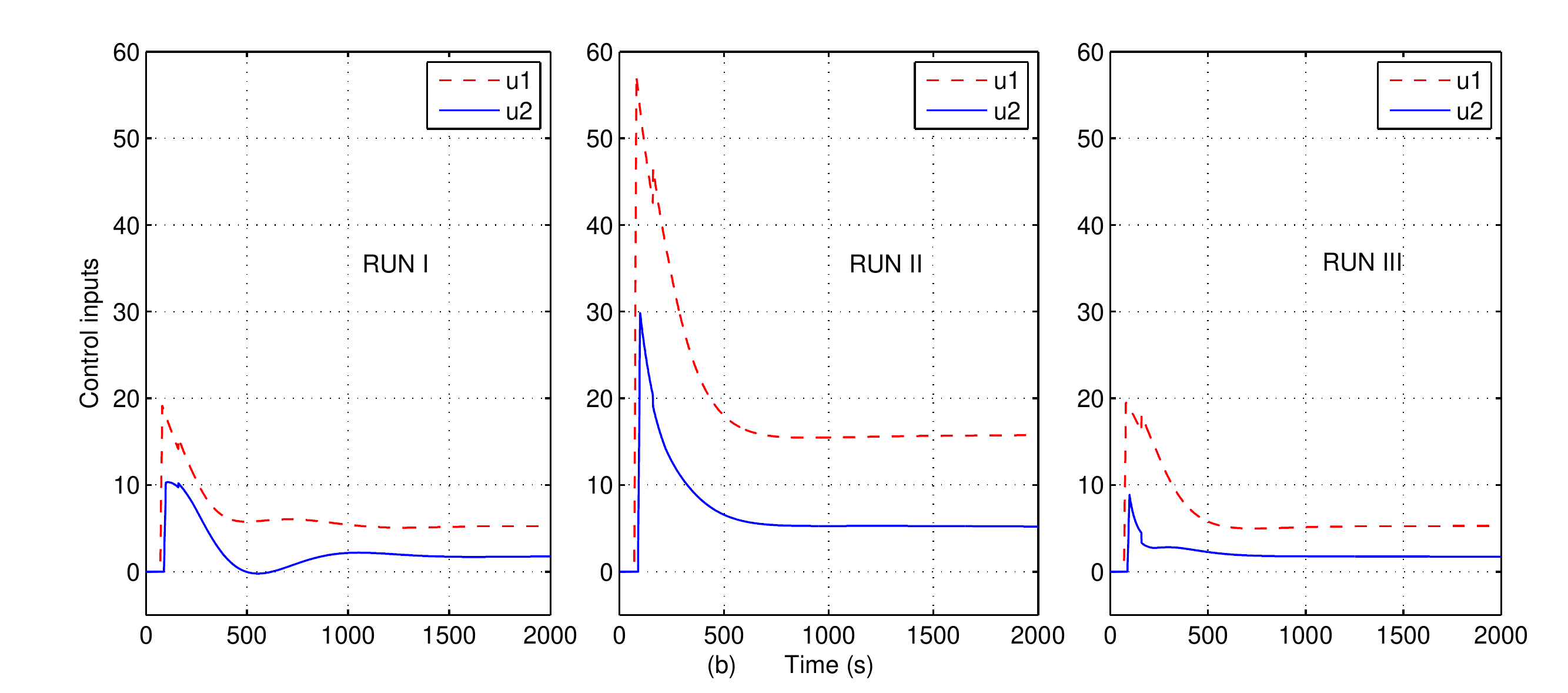}}\\
\subfloat{\includegraphics[width=\columnwidth, trim=0cm 0cm 0cm 0cm, height = 1.8in]{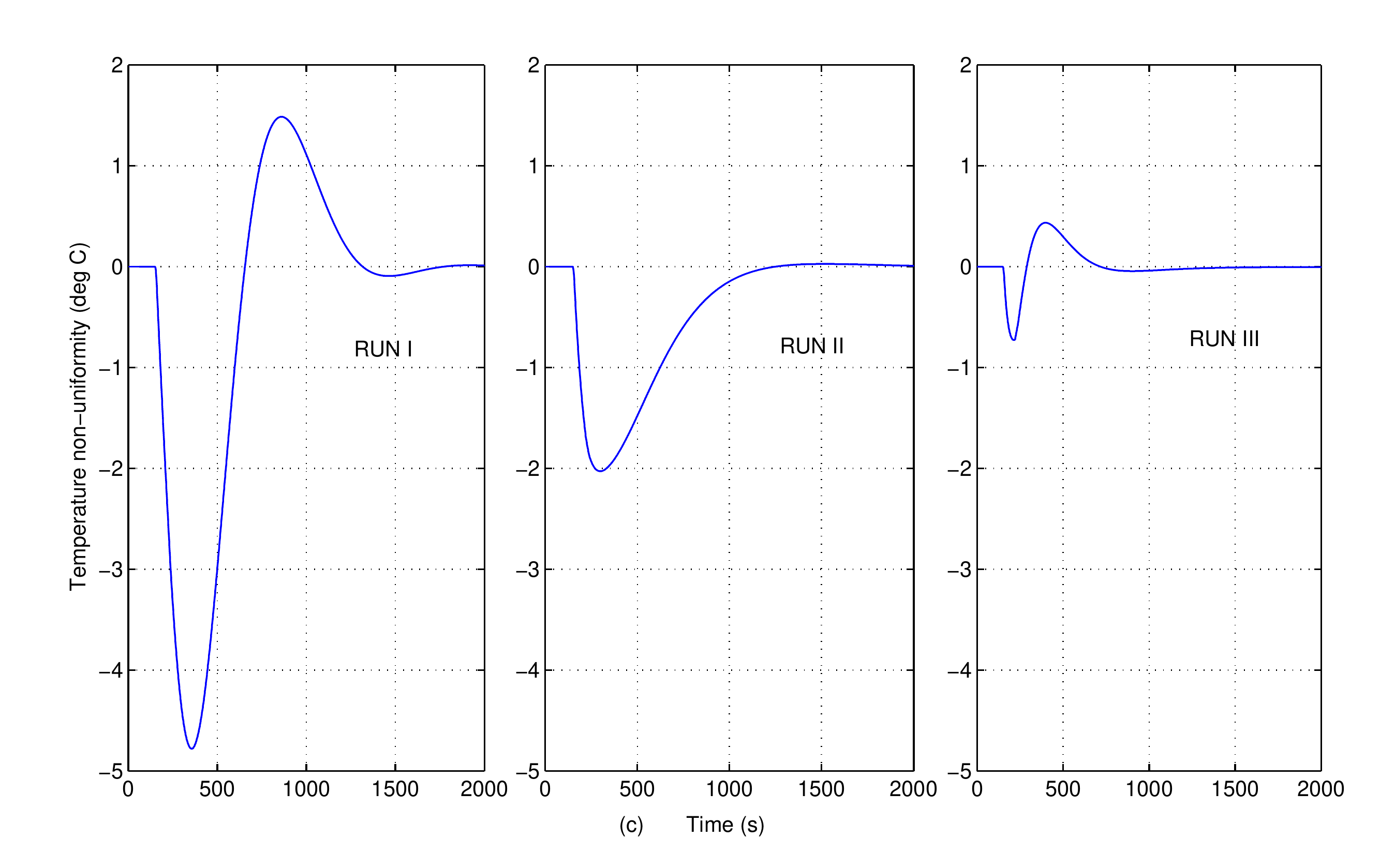}}
\caption{Comparison of (a) output response, (b) control effort and (c) temperature non-uniformity between Run I (normal predictive ratio control), Run II (setpoint variation) and Run III (ratio error cost) in the presence of a set-point change at $t = 150 s$.}
\label{c6fig4}
\end{figure}

Firstly, define the weighting matrices $Q_1$, $R_1$ as
\begin{equation}
Q_1=diag(P_1,0,I_1,P_2,0,I_2),\ R_1=\epsilon I
\end{equation}
where $I$ is an identity matrix. The ultimate gains and periods for the two processes have to be identified as $K_{u1},\ T_{u1}$ and $K_u2$, $T_u2$. Let $I_1=I_2=0$, fix the proportional gains $P_1$, $P_2$ in the $Q_1$ form above and decrease the value of $\epsilon$ until one achieves $\epsilon=min\{\epsilon \in \mathbb{R}^+ : u_1,u_2 \in \mathbb{U}, \text{no overshoot}\}$, where $P_1/P_2=(K_{u1}/K_{u2})^2$  and $\mathbb{U}$ is the input constraint set. This is also to ensure that one achieve the stability at low gain.

Increase $I_1$, $I_2$ for faster output response and desirable overshoot degree, while maintaining the ratio $I_1/I_2  =(\frac{K_{u1}}{T_{u1}}.\frac{T_{u2}}{K_{u2}})^2$. By doing this, one actually tunes $Q_1$, $R_1$ according to Ziegler-Nichols formula, but with different coefficients.

In order to tune ratio weighting parameters $\beta,\ \gamma$, it depends on the emphasis of either maximum ratio error, or fast convergence of ratio error. In general, one would increase $\beta$ to correct the response rates of the two processes, then increase $\gamma$ to possibly eliminate the remaining ratio error. This is illustrated in Fig. \ref{c6fig3}.

\section{Simulation Studies}\label{c6s5}
\subsection{Example 1}
To demonstrate the principles of the GPC-based PID scheme discussed on the previous sections, the controller is applied to maintain a ratio between two bake plate temperatures $y_1 (t),\ y_2 (t)$ of the thermal system as in [15] with input delays, represented by the process:
\begin{IEEEeqnarray}{rCl}
Y_1(s)&=&\frac{2.67e^{-60s}}{323.58s+1} U_1(s)+\frac{1.039e^{-80s}}{759.2s+1} U_2 (s)\nonumber\\
Y_2(s)&=&\frac{1.039e^{-60s}}{759.2s+1} U_1(s)+\frac{1.5595e^{-80s}}{524.5s+1} U_2 (s),
\label{c6eq23}
\end{IEEEeqnarray}
where $u_1(t),\ u_2(t)$ are the control inputs with delay $h_1=60s,\ h_2=80s$. In this example, a sampling time $t=1s$ is used. Two zone temperature changes $y_1$, $y_2$ have zero initial values, and the setpoints are $10.00^\circ\,C$. The ratio between two process variables $y_1 (t)$ and $y_2 (t)$ is kept at a tight ratio $\alpha=y_2/y_1=1.000$. 

The GPC control law is designed using prediction horizon $N=10$. Three different methods aiding ratio control to GPC-based PID are compared. The first method is the normal predictive ratio control where $r_2=\alpha r_1$, without any ratio-tightening scheme. The second method is set-point variation scheme proposed in [15] with threshold $\alpha_b=0.001$ and the gain $K=120$. The proposed method, on the other hand, considers error-ratio cost residing in performance index. The weighting parameters are chosen by the tuning procedure in Section 3.2. Here we have $Q=diag(10,0,0.007,50,0,0.1),\ R=0.6I$ and $\beta=10,\ \gamma=0.1$. The prediction horizon N is rather dependent on the calculation power, so it is chosen as $N=5$ here.

Define the output non-uniformity as $e_m=\alpha y_1-y_2$. Fig. \ref{c6fig4} shows the performance of three mentioned methods. From the output responses, one can notice that the control inputs actually react in advance to the future error which only incurs at $t=150s$. It has been also observed that the normal predictive ratio control (Run I) yields unsatisfactory results with the maximum non-uniformity of $4.81 ^\circ\,C$, as expected. The same method with setpoint variation approach (Run II) gives a relative good performance, as the uniformity is below $2.05^\circ\,C$. However, this improvement requires a very high input effort to achieve due to the different amount of process delays. For the proposed ratio error cost (Run III), the uniformity performance is better above all, and smaller control inputs are required. 

\begin{figure}[!t]
\centering
\includegraphics[width=\columnwidth]{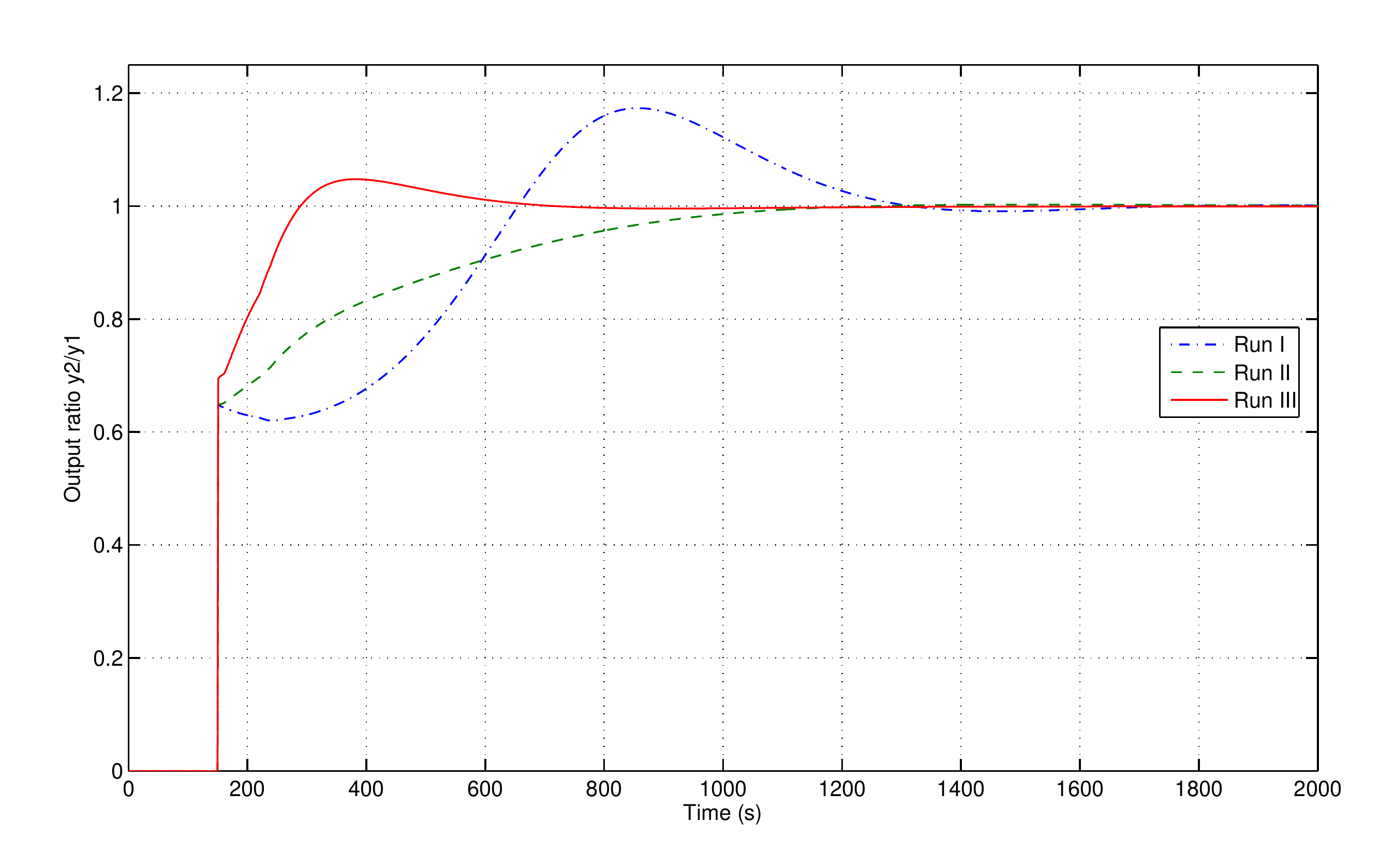}
\caption{IVF integrated platform.}
\label{c6fig5}
\end{figure}

In order to illustrate clearly the effect of the new ratio error minimization method, the actual ratio between two process variables $y_2/y_1$ in Example 1 is shown in Fig. \ref{c6fig5}. The ratio produced by the proposed method has the small deviation from the desired ratio and fast response.  

\subsection{Example 2}
In real situations, it is very difficult to identify a plant model with accurate parameters, not mentioning that the plant model may be a time-varying or non-linear system. Hence, in order to demonstrate the robustness of the suggested control scheme, parametric errors are introduced so that the real model of \eqref{c6eq24} is given by
\begin{IEEEeqnarray}{rCl}
Y_1 (s)&=&\frac{2.67xe^{-60s}}{323.58xs+1} U_1(s)+\frac{(1.039/x)e^{-80s}}{(759.2/x)s+1}U_2(s)\nonumber\\
Y_2 (s)&=&\frac{(1.039/x)e^{-60s}}{(759.2/x)s+1} U_1(s)+\frac{1.5595xe^{-80s}}{524.5xs+1}U_2(s),
\label{c6eq24}
\end{IEEEeqnarray}
with $x=1.4$ (model error up to $40\%$). 

According to the adaptive Blend station procedure, the setpoint weighting is chosen as $\gamma'=0.32$ through a series of setpoint change tests, and PI controllers are tuned by Ziegler-Nichols formula as $(k_{p1},k_{i1})=(1.514,0.016)$, $(k_{p2},k_{i2})=(3.205,$ $0.026)$. Meanwhile, the proposed controller is the same as in Example 1.

\begin{figure}
\centering
\subfloat{\includegraphics[width=\columnwidth, trim=0cm 0cm 0cm 1cm, height=1.8in]{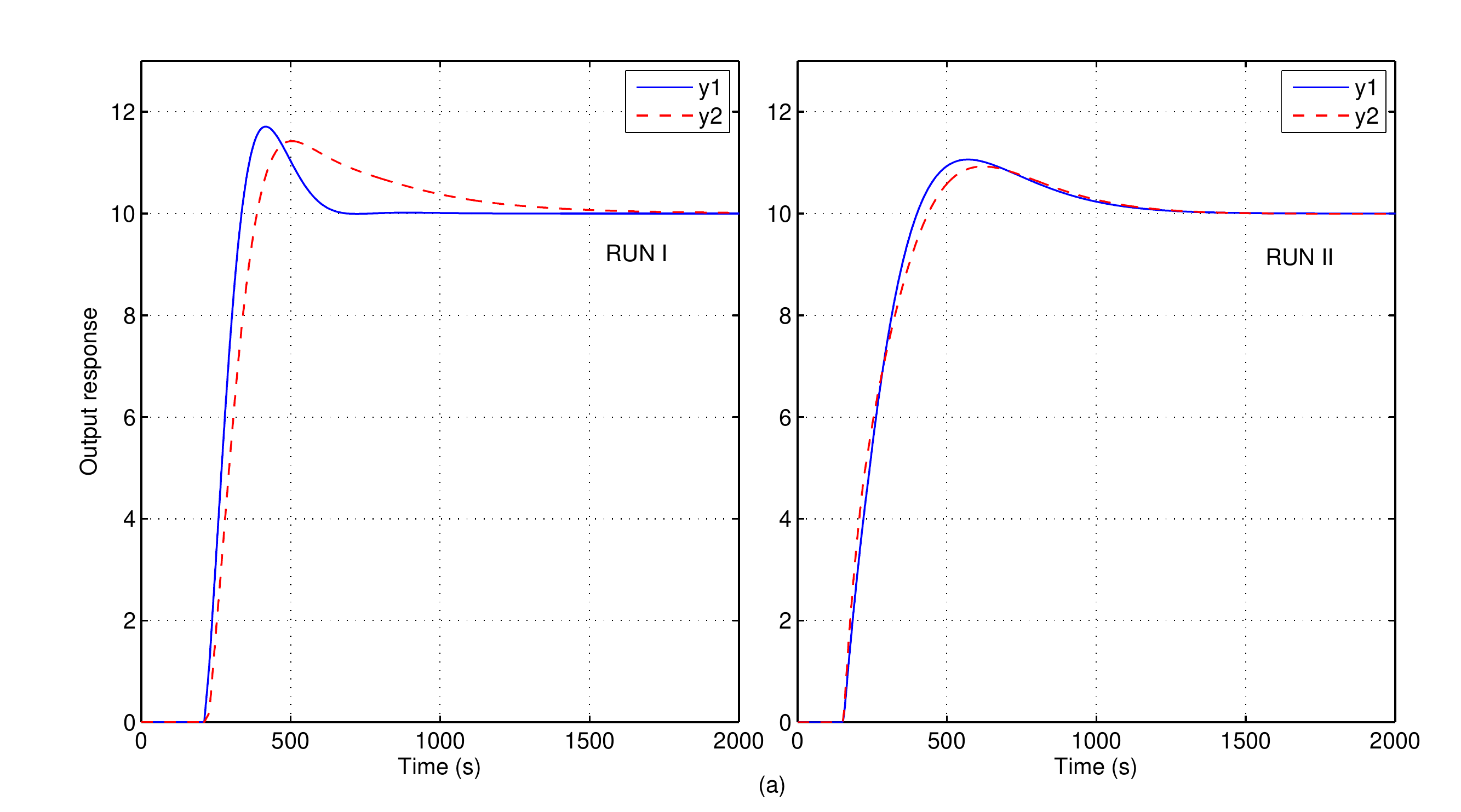}}\\
\subfloat{\includegraphics[width=\columnwidth, trim=0cm 0cm 0cm 1cm, height=1.8in]{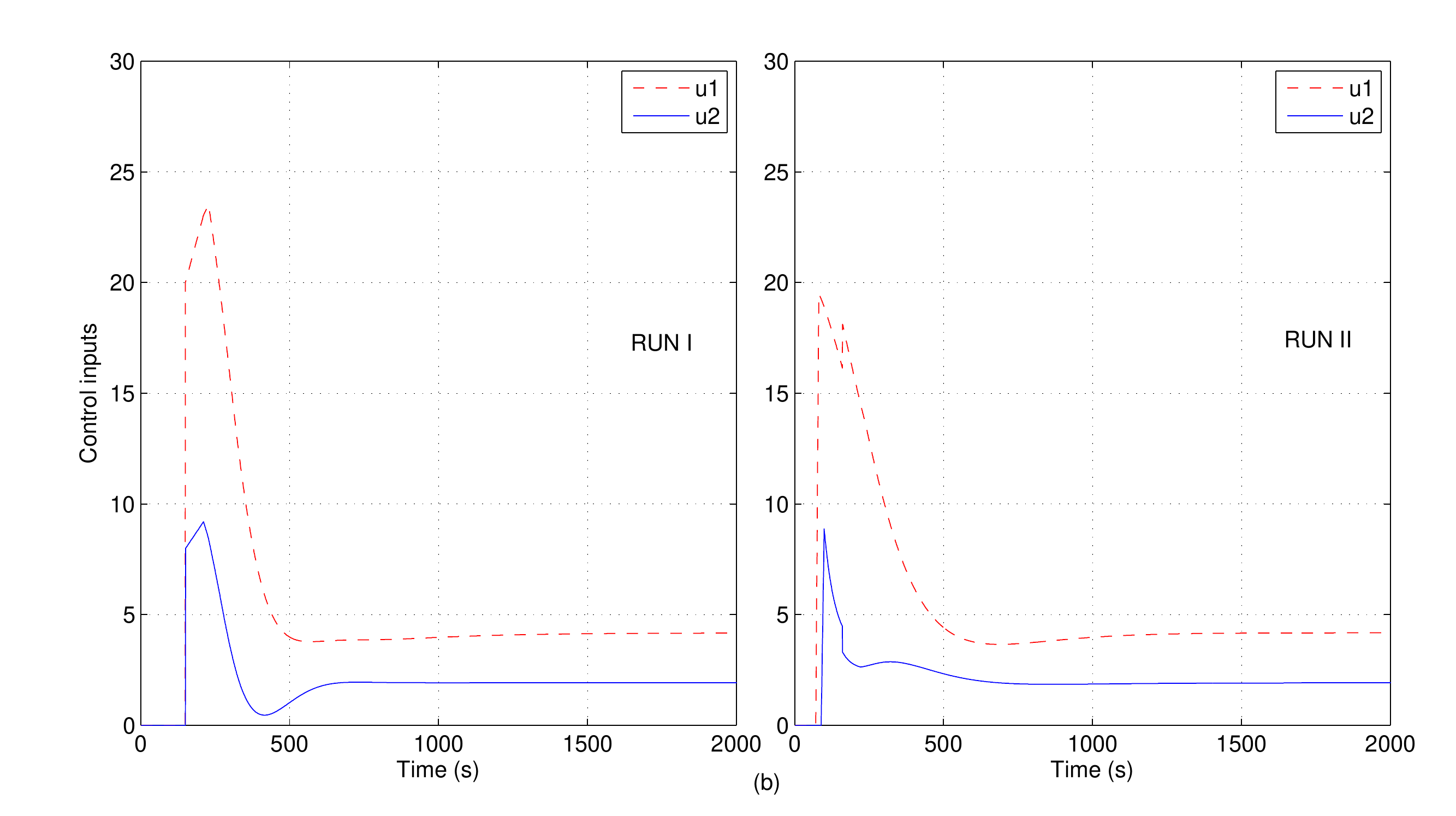}}\\
\subfloat{\includegraphics[width=\columnwidth, trim=0cm 0cm 0cm 1cm, height=1.8in]{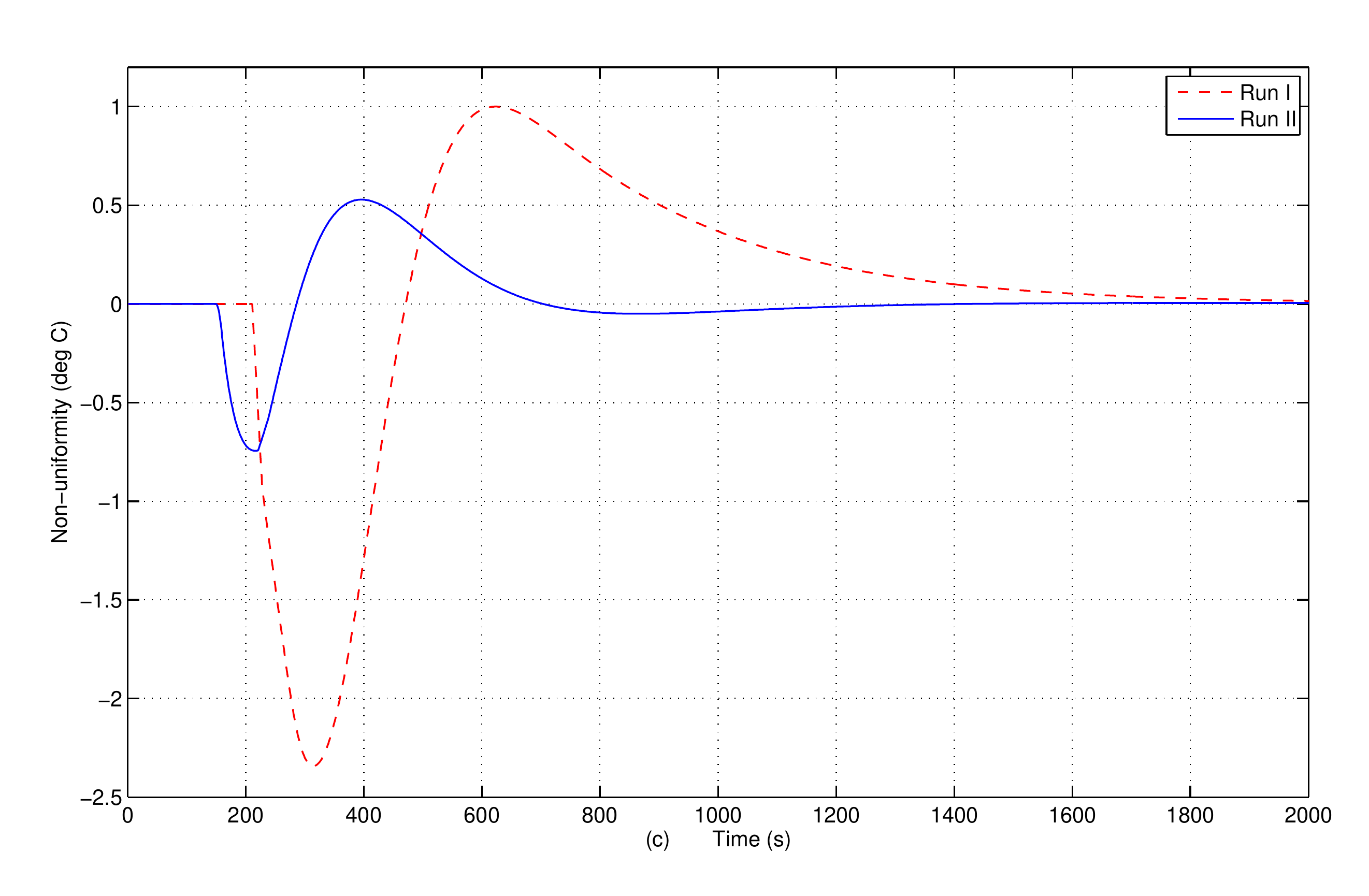}}
\caption{System response for a model perturbation in case of: (a) blend station configuration, (b) proposed method and (c) comparison in temperature non-uniformity.}
\label{c6fig6}
\end{figure}

Fig. \ref{c6fig6}a and \ref{c6fig6}b shows the output responses of the Blend station architecture in \cite{Hag01Blend} and proposed method under model errors. Fig. \ref{c6fig6}c illustrates the degree of robustness of these two schemes. The former configuration without predictive control is not able to resolve the model error and results in long recovery of ratio control. Meanwhile, the latter method recovers output non-uniformity to 0 after enduring the model mismatch. In fact, the integral cost of error ratio control suggested in Section 3 enables this flexibility as it is merged into the performance index. This may not be a proof for robust stability of the system, but it ensures that with significant model error, the proposed method still maintains its good performance.

\section{Experimental Results}\label{c6s6}
\begin{figure}[!t]
\centering
\includegraphics[width=0.75\columnwidth]{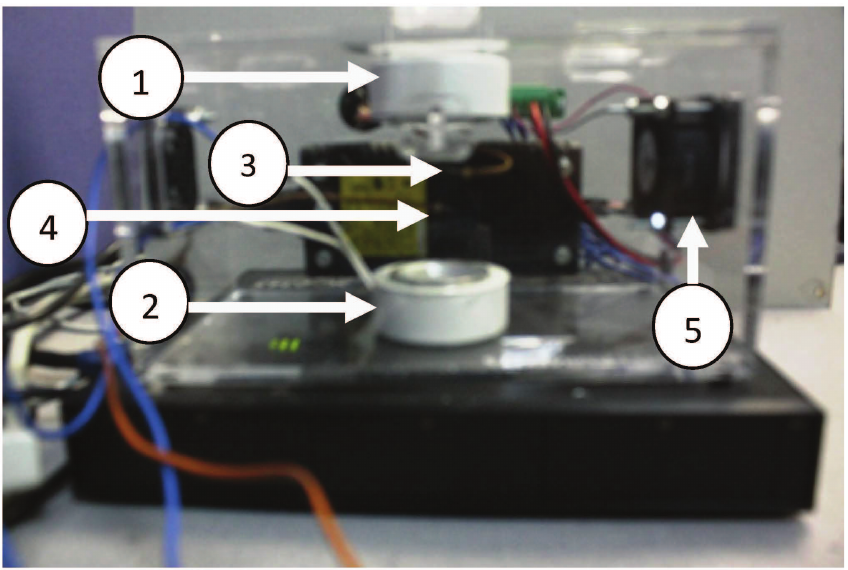}
\caption{Setup of the thermal chamber system with (1, 2) J-Type thermocouples, (3,4) halogen lights and (5) cooling fan.}
\label{c6fig7}
\end{figure}

Fig. \ref{c6fig7} presents the setup of a desktop thermal chamber, mounted on a National Instrument (NI) SC-2345 platform with configurable connectors. In this real-time experiment, the air temperature can be controlled by adjusting the power of the lights and the fan. The variables of interest are the air temperatures $y_1,y_2$ sensed by SCC-TC02 J-type thermocouples at two different height locations. These outputs are manipulated through the upper and lower halogen bulbs in an interactive process. Different delays $h_1,h_2$ are contained in the two input channels. Besides, the cooling fan fulfills the role of disturbance source. NI LABVIEW is used to develop a controller for this system.

One can reasonably assume the above system as a nonlinear process, due to the advection of air. In this experiment, simple system identification through step responses is exploited in a particular operating point to estimate and formulate a first-order system with delays, as follows:
\begin{IEEEeqnarray}{rCl}
Y_1 (s)=\frac{35e^{-2s}}{51s+1} U_1(s)+\frac{25.5e^{-6s}}{99s+1} U_2(s)\nonumber\\
Y_2 (s)=\frac{19e^{-2s}}{108s+1} U_1 (s)+\frac{31.5e^{-6s}}{68s+1} U_2 (s),
\label{c6eq25}
\end{IEEEeqnarray}
Initial values of the two outputs are $y_{1o}=y_{2o}=26^\circ\,C$. A setpoint change of $5^o\,C$ is given for the first output $y_1$, and the ratio $alpha=y_2/y_1=1.000$ is to be maintained during the process. In addition, notice that the input constraint is present here, whereby $0\leq u_1,u_2 \leq 1$. The sampling rate is $0.1s$.

The objective of this experiment is to show how the MPC implementation with ratio control can cope with this interactive system when compared with a fixed PID regulator. In this experiment, besides the potential model error, there is also a disturbance to test the performance of these two methods. Again the fixed PID regulator was chosen as Blend station design tuned to provide good ratio control of the given process with fast response and no excessive overshoot: $\gamma'=0.75$ and $(k_{p1},k_{i1})=(0.31,0.045),\ (k_p2,k_i2 )=(0.07,0.0036)$. The parameters of predictive PID ratio control were adjusted through the tuning procedure provided in Section 3.2. The prediction horizon is given as $N=5$. $Q=diag(1,0,0.001,1,0,0.001),\ R=5I$ and $\beta=5,\ \alpha=0.15$.

\begin{figure}
\centering
\subfloat[]{\includegraphics[width=\columnwidth]{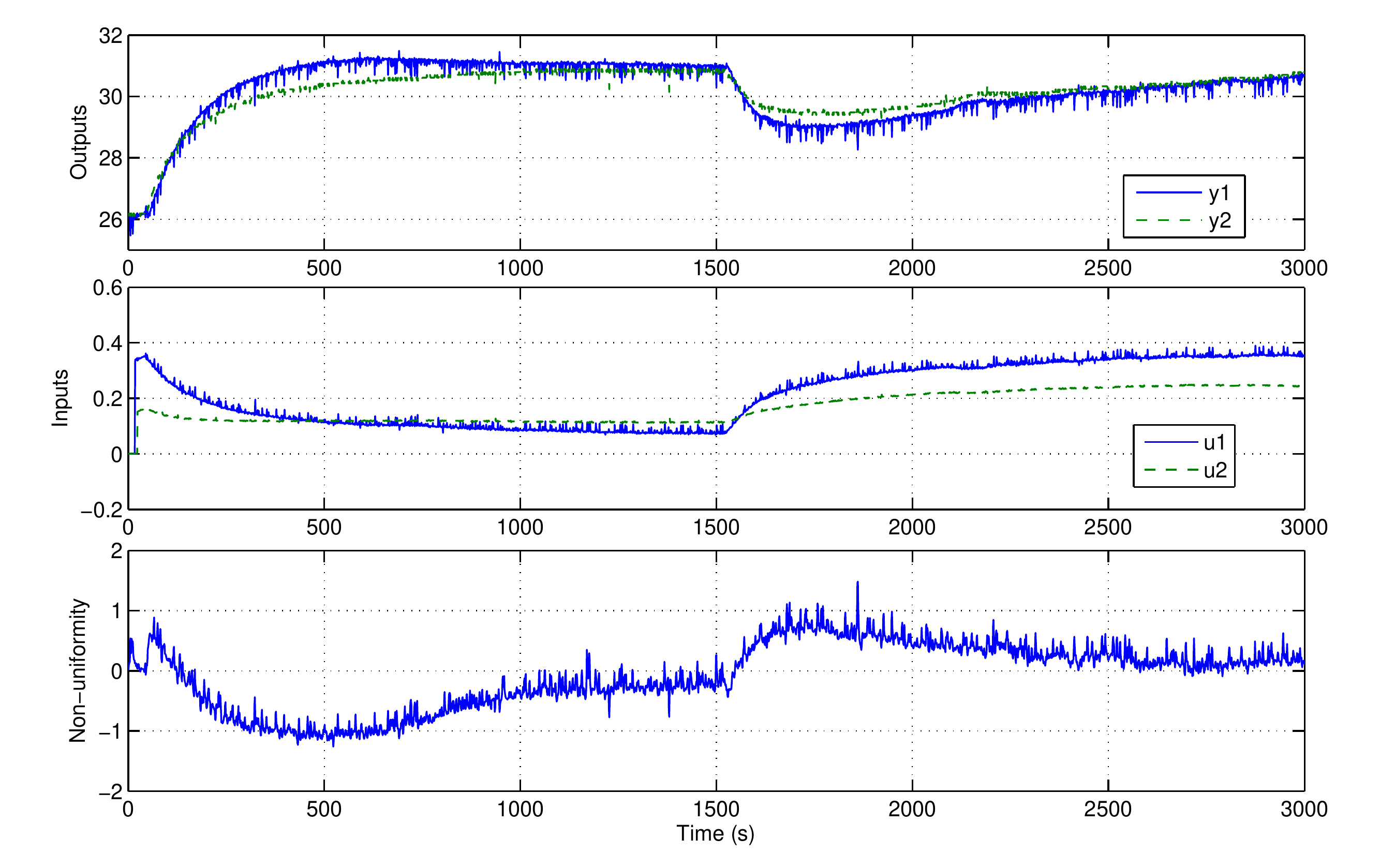}}\\
\subfloat[]{\includegraphics[width=\columnwidth]{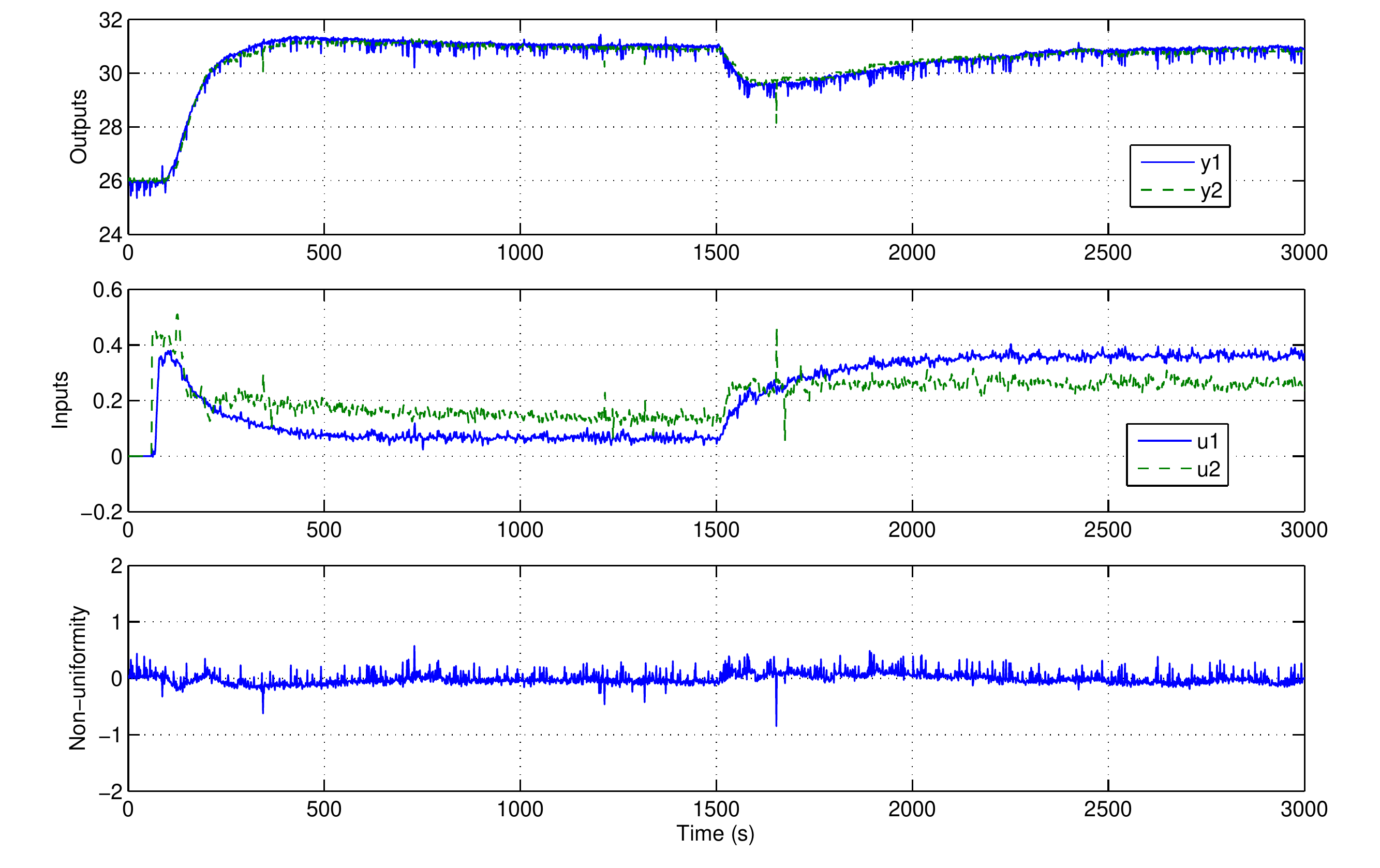}}
\caption{Performance of (a) the Blend station PID and (b) the predictive PID controller.}
\label{c6fig8}
\end{figure}

Fig. \ref{c6fig8}a shows the behavior of a fixed PID regulator, giving a reasonable but rather non-uniform control and so, a poor ratio performance. Due to the interacting feature in the processes, the response rates are different. The control inputs $u_1$, $u_2$ only respond to the output errors individually. The same situation happens in the event of an unpredicted step disturbance $d=1$.

\begin{table}[!t]
\caption{Non-uniformity statistics for output ratio control in the thermal interaction experiment.}
\centering
\setlength{\tabcolsep}{12pt}
\begin{tabular}{llll}
\hline
Controller	&Abs. Peak	& Mean	& RMS\\ \hline
Parallel PID & 1.1918 & 0.0764 & 0.549\\
Predictive PID	& 0.2693 & 0.0065	& 0.102 \\ \hline
\end{tabular}
\setlength{\tabcolsep}{5pt}
\label{c6tab1}
\end{table}

The simulation results using the proposed MPC ratio control with variable PID gains are shown in Fig. \ref{c6fig8}b. The rates, as well as the shapes, of output response are closely followed. Moreover, recovery after disturbance is also faster, along with the uniformity of the outputs. This can be attributed to the corporation between the control inputs during the course of transient response. Performance statistics are shown in Table \ref{c6tab1}, with the absolute peak, mean and root-mean-square of the ratio non-uniformity are considered. The proposed method helps improve the performance from five to ten times, according to the data.

\section{Conclusion}\label{c6s7}
This chapter presents a predictive feed-forward and PID control scheme based on MPC that copes with ratio control for interacting delayed processes. Compared to the standard parallel configuration, the proposed method allows one to take into account the ratio error cost, thus tightening the output ratio towards desired value. In addition, this method is more efficient than the mentioned approaches as the dynamic ratio error information improves the optimal control input through PID gains instead of feed-forward calculation. With the new ratio control scheme, a better performance in output ratio control is achieved with smaller control effort.

\bibliographystyle{elsarticle-harv}
\bibliography{D:/Dropbox/Research/referencelist}
\end{document}